\documentclass[runningheads]{llncs}

\usepackage{llncsdoc}
\usepackage{epsfig}
\usepackage{graphicx}
\usepackage[leqno]{amsmath}
\usepackage{amssymb}
\usepackage{amsmath}
\usepackage{psfrag}
\usepackage{rotating}
\usepackage[usenames]{color}
\usepackage{multirow}

\newcommand{\R}{\mathbb{R}}

\newcommand{\p}{\varphi}

\newcommand{\eps}{\varepsilon}
\newcommand{\dgm}{\textrm{Dgm}}
\newcommand{\cost}{\mathrm{cost}}\newcommand{\newatop}[2]{\genfrac{}{}{0pt}{1}{#1}{#2}}

\newcommand\restr[2]{{% we make the whole thing an ordinary symbol
  \left.\kern-\nulldelimiterspace % automatically resize the bar with \right
  #1 % the function
  \vphantom{\big|} % pretend it's a little taller at normal size
  \right|_{#2} % this is the delimiter
  }}

\begin{document}

\pagestyle{headings}

\mainmatter

\title{The coherent matching distance\\ in 2D persistent homology\thanks{Work carried out under the auspices of INdAM-GNSAGA. M.E. has been partially supported by the Toposys project FP7-ICT-318493-STREP, as well as an ESF Short Visit grant under the Applied and Computational Algebraic Topology networking programme. A.C. is partially supported by the FP7 Integrated Project IQmulus, FP7-ICT-2011-318787, and the H2020 Project Gravitate, H2020 - REFLECTIVE - 7 - 2014 - 665155.}}

\titlerunning{The coherent matching distance in 2D persistent homology} 

\authorrunning{A. Cerri, M. Ethier and P. Frosini} 

\author{Andrea Cerri\inst{1} \and Marc Ethier\inst{2,3} \and Patrizio Frosini\inst{4}}

\institute{IMATI -- CNR, Genova, Italia\\
\email{andrea.cerri@ge.imati.cnr.it}
\and
Facult\'e des Sciences, Universit\'e de Saint-Boniface, Winnipeg, Manitoba, Canada\\
\email{methier@ustboniface.ca}
\and
Institute of Computer Science and Computational Mathematics, Jagiellonian University, Krak\`ow, Poland
\and
Dipartimento di Matematica, Universit\`a di Bologna, Italia\\
\email{patrizio.frosini@unibo.it}
}

\maketitle             

\begin{abstract}
Comparison between multidimensional persistent Betti numbers is often based on the multidimensional matching distance. While this metric is rather simple to define and compute by considering a suitable family of filtering functions associated with lines having a positive slope, it has two main drawbacks. First, it forgets the natural link between the homological properties of filtrations associated with lines that are close to each other. As a consequence, part of the interesting homological information is lost. Second, its intrinsically discontinuous definition makes it difficult to study its properties. In this paper we introduce a new matching distance for 2D persistent Betti numbers, called \emph{coherent matching distance} and based on matchings that change coherently with the filtrations we take into account. Its definition is not trivial, as it must face the presence of monodromy in multidimensional persistence, i.e. the fact that different paths in the space parameterizing the above filtrations can induce different matchings between the associated persistent diagrams. In our paper we prove that the coherent 2D matching distance is well-defined and stable.\\

\noindent {\bf Keywords:} Multidimensional matching distance, multidimensional persistent Betti numbers, monodromy
\end{abstract}

\section*{Introduction}
In the last twenty-five years the concept of \emph{topological persistence} has become of common use in computational geometry and topological data analysis. It is based on the idea that the most important properties of a filtered topological space are the ones that persist under large changes of the parameters defining the sublevel sets in the filtration. The concept of persistence revealed quite useful in extracting information from data that can be described by $\R^h$-valued functions defined on a topological space (e.g., images or point clouds representing 3D-models, via a distance function). 

The theory of topological persistence was initially developed for the case $h=1$, but in recent years the interest in the case $h>1$ has rapidly increased, leading to new theoretical developments and computational methods (cf., e.g., \cite{BiCe*11,CaSiZo10,CaZo09,Le15}). One of this methods is based on a reduction of the $h$-dimensional case to the $1$-dimensional setting by using a suitable family of derived real-valued functions \cite{BiCe*08,CaDFFe10,CeDi*13}. If $h=2$, it consists of changing the 2D filtration given by a filtering function $f=(f_1,f_2):M\to \R^2$ into the 1D filtrations associated with the real-valued functions $f_{a,b}^*:M\to\R$ defined as $f_{a,b}^*(x):=\min\{a,1-a\}\cdot\max\left\{\frac{f_1(x)-b}{a},\frac{f_2(x)+b}{1-a}\right\}$, for $a\in ]0,1[$ and $b\in\R$. This approach allows for introducing a distance $D_{match}(\beta_f,\beta_g)$ between the persistent Betti numbers associated with $f$ and $g$, which is defined as the supremum of the classical bottleneck distance between the persistent diagrams of $f_{a,b}^*$ and $g_{a,b}^*$, varying $a$ and $b$.

While this method brings back the problem to the 1D case, it opens the way to new issues of interest. First of all, the distance $D_{match}$ forgets the natural link between the homological properties of filtrations associated with pairs $(a,b)$ that are close to each other. This fact implies that part of the homological information is lost. Second, its intrinsically discontinuous definition makes it difficult to study its properties. 

As a possible answer to these observations, we introduce in this paper a new matching distance for 2D persistent Betti numbers, called \emph{coherent matching distance} and based on the use of matchings that change continuously with respect to the filtrations we take into account. In order to state its definition, we have to manage the problem of monodromy, consisting of the fact that a loop in the space parameterizing the above filtrations can induce a transformation that changes a matching $\sigma$ between the associated persistent diagrams into a matching $\tau\neq\sigma$~\cite{CeEtFr13}. 

The paper is organized as follows. In Section~\ref{MS}, we recall the definitions of multidimensional persistent Betti number and multidimensional matching distance, together with the monodromy phenomenon in 2D persistent homology. In Section~\ref{CMD}, we introduce  the coherent 2D matching distance and prove that it is well-defined and stable. 
%In particular, in Section~\ref{DvCD} we show that  the new distance is strictly more informative than the old one.

\section{Mathematical setting}\label{MS}

Let $f=(f_1,f_2)$ be a continuous map from a finitely triangulable topological space $M$ to the real plane $\R^2$.

\subsection{Persistent Betti numbers}\label{PBN}

As a reference for multidimensional persistent Betti numbers we use \cite{CeDi*13}. According to the main topic of this paper, we will also stick to the notations and working assumptions adopted in \cite{CeEtFr13}. In particular, we build on the strategy adopted in the latter to study certain instances of monodromy for multidimensional persistent Betti numbers. Roughly, the idea is to reduce the problem to the analysis of a collection of persistent Betti numbers associated with a real-valued function, and their compact representation in terms of \emph{persistence diagrams}.

We use the following notations: $\Delta^+$ is the open set $\{(u,v)\in\R\times\R:u< v\}$. $\Delta$ represents the diagonal set $\{(u,v)\in\R\times\R:u= v\}$. We can further extend $\Delta^+$ with points at infinity of the kind $(u,\infty)$, where $|u|<\infty$. Denote this set $\Delta^*$. For a continuous function $\p:M\to\R$, and for any $n\in\mathbb{N}$, if $u<v$, the inclusion map of the sublevel set $M_u=\{x\in M:\p(x)\leq u\}$ into the sublevel set $M_v=\{x\in M:\varphi(x)\leq v\}$ induces a homomorphism from the $n$th homology group of $M_u$ into the $n$th homology group of $M_v$. The image of this homomorphism is called the {\em $n$th persistent homology group of $(M,\p)$ at $(u,v)$}, and is denoted by $H_n^{(u,v)}(M,\p)$. In other words, the group $H_n^{(u,v)}(M,\p)$ contains all and only the homology classes of $n$-cycles born before or at $u$ and still alive at $v$. 

Following \cite{CeDi*13}, we assume the use of \v{C}ech homology, and refer the reader to that paper for a detailed explanation about preferring this homology theory to others. Also, we work with coefficients in a field $\mathbb{K}$, so that homology groups are vector spaces. Therefore, they can be completely described by their dimension, leading to the following definition \cite{EdLeZo02}.

\begin{definition}[Persistent Betti Numbers]\label{Rank}
The {\em persistent Betti numbers function} of $\p$, briefly PBN, is the function $\beta_{\p}:\Delta^+\to\mathbb{N}\cup\{\infty\}$ defined by
\begin{displaymath}
\beta_{\p}(u,v)=\dim H_n^{(u,v)}(M,\p).
\end{displaymath}
\end{definition}
Under the above assumptions for $M$, it is possible to show that $\beta_{\p}$ is finite for all $(u,v)\in\Delta^+$ \cite{CeDi*13}. Obviously, for each $n\in\mathbb{Z}$, we have different PBNs of $\p$ (which might be denoted by $\beta_{\p,n}$, say), but for the sake of notational simplicity we omit adding any reference to $n$.

The PBNs of $\p$ can be simply and compactly described by the corresponding \emph{persistence diagrams}. Formally, a persistence diagram can be defined via the notion of \emph{multiplicity} \cite{CoEdHa07,FrLa01}. Following the convention used for PBNs, any reference to $n$ will be dropped in the sequel.
\begin{definition}[Multiplicity]\label{Multiplicity}
The \emph{multiplicity} $\mu_{\p}(u,v)$  of $(u,v)\in\Delta^+$ is the finite, non-negative number given by 
\begin{equation*}
\min_{\newatop{\eps>0}{u+\eps<v-\eps}} \beta_{\p}(u+\eps ,v-\eps)-\beta_{\p}(u-\eps ,v-\eps)-\beta_{\p}(u+\eps,v+\eps)+\beta_{\p}(u-\eps ,v+\eps).
\end{equation*}
The \emph{multiplicity} $\mu_{\varphi}(u,\infty)$  of $(u,\infty)$ is the finite, non-negative number given by {\setlength\arraycolsep{1pt}
\begin{equation*}
\min_{\eps > 0,\,u+\eps<v} \beta_{\p}(u+\eps,v)-\beta_{\p}(u-\eps ,v).
\end{equation*}}
\end{definition}
\begin{definition}[Persistence Diagram]\label{persDiag}
The persistence diagram $\dgm(\p)$ is the multiset of all points $(u,v)\in\Delta^*$ such that $\mu_{\p}(u,v)>0$, counted with their multiplicity, union the points of $\Delta$, counted with infinite multiplicity.
\end{definition}
Each point $(u,v)\in \Delta^*$ with positive multiplicity will be called a \emph{cornerpoint}. A cornerpoint $(u,v)$ will be said to be a \emph{proper cornerpoint} if $(u,v)\in \Delta^+$, and a \emph{cornerpoint at infinity} if $(u,v)\in \Delta^*\setminus \Delta^+$.

\subsection{2-dimensional setting} 
The definition of persistent Betti numbers can be easily extended to $\R^h$-valued functions \cite{CeDi*13}. It has been proved that, in this case, the information enclosed in the persistent Betti numbers is equivalent to that represented  by the set of persistent Betti numbers associated with a certain family of real-valued functions. We discuss this for the specific case of the above function $f:M\to\R^2$, referring the reader to Figure~\ref{foliation} for a pictorial representation.

\begin{figure}[ht]
\begin{center}
\psfrag{SS}{\tiny $s(a,1-a)+(b,-b)$}
\psfrag{TT}{{\tiny $t(a,1-a)+(b,-b)$}}
\psfrag{ro}{$\dgm(f_{a,b}^*)$}
\psfrag{(b,-b)}{$(b,-b)$}
\psfrag{b1+b2=0}{$u+v=0$}
\psfrag{s}{$s$}\psfrag{t}{$t$}
\psfrag{(a,b)}{$(a,b)$}
\psfrag{f1}{$f_1$}
\psfrag{f2}{$f_2$}
\psfrag{u}{$u$}\psfrag{v}{$v$}
\psfrag{X}{$r_{a,b}$}
\psfrag{l}{$(a,1-a)$}
\includegraphics[width=0.9\textwidth]{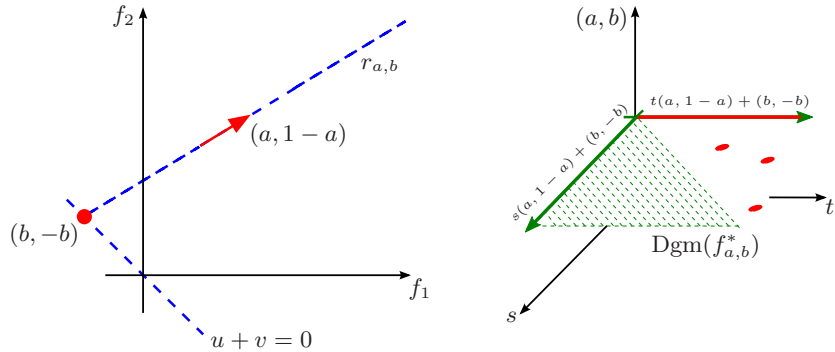}
\end{center}
\caption{Correspondence between an admissible line $r_{a,b}$
and the persistence diagram $\dgm(f_{a,b}^*)$}
\label{foliation}
\end{figure}

Consider the pairs $(a,b)\in\,]0,1[\,\times\R$. Any such pair identifies an oriented line $r_{a,b}\in\R^2$ of positive slope, parameterized by $t$ with equation $(u,v)=t\cdot (a,1-a)+(b,-b)$. The space $\Lambda$ of lines obtained according to this procedure is referred to as the \emph{set of admissible lines}, whereas $P(\Lambda)$ denotes the set of pairs $(a,b)$ parameterizing $\Lambda$. The generic point $(u,v)=t\cdot (a,1-a)+(b,-b)$ of $r_{a,b}$ can be associated with the sublevel set of $M$ defined as $\{x\in M: f_1(x)\leq u,\,f_2(x)\leq v\}$, which is equivalent to that given by $\{x\in M: f_{a,b}(x)\leq t\}$ induced by the real-valued function $f_{a,b}:M\to\R$ with $f_{a,b}(x):=\max\left\{\frac{f_1(x)-b}{a},\frac{f_2(x)+b}{1-a}\right\}$. In this setting, the Reduction Theorem proved in \cite{CeDi*13} states that the persistent Betti numbers $\beta_{f}$ can be completely recovered by considering all and only the persistent Betti numbers $\beta_{f_{a,b}}$ associated with the admissible lines $r_{a,b}$, which are in turn encoded in the corresponding persistence diagrams $\dgm(f_{a,b})$.

\subsubsection{2-dimensional matching distance}\label{2DMD}

Assume now that we have two continuous functions $f,g:M\to\R^2$. We consider the persistence diagrams $\dgm(f_{a,b})$, $\dgm(g_{a,b})$ associated with the admissible line $r_{a,b}$, and normalize them by multiplying their points by $\min\{a,1-a\}$. This is equivalent to consider the \emph{normalized persistence} diagrams $\dgm(f_{a,b}^*)$, $\dgm(g_{a,b}^*)$, with $f_{a,b}^*=\min\{a,1-a\}\cdot f_{a,b}$ and $g_{a,b}^*=\min\{a,1-a\}\cdot g_{a,b}$, respectively. The 2-dimensional matching distance $D_{match}(\beta_f,\beta_g)$ \cite{BiCe*08} is then defined as 
$$
D_{match}(\beta_f,\beta_g)=\sup_{P(\Lambda)}d_B(\dgm(f_{a,b}^*),\dgm(g_{a,b}^*)),
$$ 
with $d_B(\dgm(f_{a,b}^*),\dgm(g_{a,b}^*))$ denoting the usual bottleneck distance between the normalized persistence diagrams $\dgm(f_{a,b}^*)$ and $\dgm(g_{a,b}^*)$ \cite{CoEdHa07,DaFrLa10}.

\begin{remark}\label{remNormStability}
The introduction of normalized persistence diagrams is crucial here. Indeed, the bottleneck distance $d_B(\dgm(f_{a,b}^*),\dgm(g_{a,b}^*))$ is stable against functions' perturbations when measured by the $\sup$-norm, while this is not true for the distance $d_B(\dgm(f_{a,b}),\dgm(g_{a,b}))$, see \cite[Thm.~4.2]{CeDi*13} for details.
\end{remark}

\subsubsection{Monodromy in 2-dimensional persistent homology}
We know that normalized persistence diagrams are stable with respect to changes of the underlying functions, when the $\sup$-norm is considered (Remark~\ref{remNormStability}). Since each function $f_{a,b}^*$ depends continuously on the parameters $a$ and $b$ with respect to the $\sup$-norm, it follows that the set of points in $\dgm(f_{a,b}^*)$ depends continuously on the parameters $a$ and $b$. Analogously, the set of points in $\dgm(g_{a,b}^*)$ depends continuously on the parameters $a$ and $b$. Suppose that $\sigma_{a,b}:\dgm(f_{a,b}^*)\to \dgm(g_{a,b}^*)$ is an \emph{optimal matching}, i.e. one of the matchings achieving the bottleneck distance $d_B(\dgm(f_{a,b}^*),\dgm(g_{a,b}^*))$. Given the above arguments, a natural question arises, whether $\sigma_{a,b}$ changes continuously varying $a$ and $b$. In other words, we wonder if it is possible to straighforwardly introduce a notion of \emph{coherence} for optimal matchings with respect to the elements of $P(\Lambda)$.

Perhaps surprisingly, the answer is no. A first obstruction is given by the fact that, trying to continuously extend a matching $\sigma_{a,b}$, the identity of points in the (normalized) persistent diagrams is not preserved when considering an admissible pair $(\bar a,\bar b)$ for which either $\dgm(f_{\bar a,\bar b}^*)$ or $\dgm(g_{\bar a,\bar b}^*)$ has points with multiplicity greater than 1. In other words, we cannot follow the path of a cornerpoint when it collides with another cornerpoint. On the one hand, this problem can be solved by replacing $P(\Lambda)$ with its subset $\textrm{Reg}(f)\cap \textrm{Reg}(g)$, with 
{\setlength\arraycolsep{2pt}
\begin{eqnarray*}
\textrm{Reg}(f) & = & \{(a,b)\in P(\Lambda) | \textrm{$\dgm(f_{a,b}^*)$  does not contain multiple points}\},\\
\textrm{Reg}(g) & = & \{(a,b)\in P(\Lambda) | \textrm{$\dgm(g_{a,b}^*)$  does not contain multiple points}\}.
\end{eqnarray*}
Throughout the rest of the paper, we will talk about \emph{singular pairs for $f$} to denote the pairs $(a,b)\in P(\Lambda)\setminus \textrm{Reg}(f)$, and about  \emph{regular pairs for $f$} to denote the pairs $(a,b)\in \textrm{Reg}(f)$. An analogous convention holds referring to the singular and regular pairs for $g$. 

On the other hand, however, continuously extending a matching $\sigma_{a,b}$ presents some problems even in this setting. Roughly, the process of extending $\sigma_{a,b}$ along a path $c:[0,1]\to\textrm{Reg}(f)\cap\textrm{Reg}(g)$ depends on the homotopy class of $c$ relative to its endpoints. This phenomenon is referred to as \emph{monodromy in 2-dimensional persistent homology}, and has been studied for the first time in \cite{CeEtFr13}.} In what follows we will show how to overcome this issue in order to define a \emph{coherent} modification of the standard 2-dimensional matching distance $D_{match}$.

\subsection{Working assumptions}\label{workAss}
To simplify the exposition, in what follows we state our results by assuming that $M$ is homeomorphic to $S^{m}$, with $m\geq 2$. In particular, this implies that all normalized persistence diagrams $\dgm(f_{a,b}^*)$, $\dgm(g_{a,b}^*)$ contain a single cornerpoint at infinity in degree 0 and $n$, and no cornerpoint at infinity in the other homology degrees. In this way, the problem of continuously extending a matching can be restricted to considering only proper cornerpoints, as there are no ambiguities in following the evolution of cornerpoints at infinity. Also, we assume that 
\begin{enumerate}
\item the functions $f,g:M\to\R^2$ are \emph{normal}, i.e. the sets of singular pairs for $f$ and $g$ are discrete \cite{CeEtFr13};
\item a constant real value $k>0$ exists such that if two proper cornerpoints $X_1,X_2$ of $\dgm^*(f_{a,b})$ have Euclidean distance less than $k$ from $\Delta$, then the Euclidean distance between $X_1$ and $X_2$ is not smaller than $k$, for all $(a,b)\in P(\Lambda)$. The same property holds for $\dgm^*(g_{a,b})$.
%\item the closure of the set $\bigcup_{a,b}(\dgm(f^*_{a,b}))\setminus\Delta$ meets $\Delta$ at a finite set of points. The same property holds for the closure of the set $\bigcup_{a,b}(\dgm(g^*_{a,b}))\setminus\Delta$.
\end{enumerate}

\section{The coherent 2-dimensional matching distance}\label{CMD}
The existence of monodromy implies that each loop in $\textrm{Reg}(f)$ induces a permutation on $\dgm(f_{a,b}^*)$. In other words, it is not possible to establish which point in $\dgm(f_{a,b}^*)$ corresponds to which point in $\dgm(f_{a',b'}^*)$ for $(a,b)\neq (a',b')$, since the answer depends on the path that is considered from $(a,b)$ to $(a',b')$ in the parameter space $\textrm{Reg}(f)$. As a consequence, different paths going from $(a,b)$ to $(a',b')$ might produce different results while extending a matching $\sigma_{a,b}$. However, it is still possible to define a notion of coherent 2-dimensional matching distance.

\subsection{Transporting a matching along a path}\label{T}
First, we need to specify the concept of transporting a proper cornerpoint $X\in\dgm(\varphi)$ along a homotopy $h(\tau,x) := (1-\tau)\cdot\varphi(x) + \tau\cdot\psi(x)$, with $\varphi,\psi:M\to\R$.

\begin{definition}[Admissible path]\label{admPath}
Let $p\in [0,1]$. A continuous path $P:[0,p]\to\Delta^+\cup\Delta$ is said to be \emph{admissible for $h$ at $\bar p\in [0,p]$} if the following hold:
\begin{enumerate}
\item $P(\tau)\in\dgm(h(\tau,\cdot))$ for $\tau\in [0,p]$;
\item $P([0,p])\cap \Delta$ is finite; 
\item if $P(\bar p)\in\Delta$ then there is no $p'\in ]\bar p, p]$ such that $P([\bar p,p'])=\{P(\bar p)\}$ and a continuous path $Q:[\bar p,p']\to\Delta^+\cup\Delta$ exists for which $Q(\tau)\in\dgm(h(\tau,\cdot))$ for $\tau\in [\bar p, p']$, $Q(\bar p)=P(\bar p)$ and $Q([\bar p,p'])\neq\{P(\bar p)\}$.
\end{enumerate}
P is called admissible for $h$ if it is admissible for $h$ at every point of its domain.
\end{definition}
In other words, $P$ is not admissible for a homotopy $h$ if it ``stops'' at a point $P(\bar p)\in\Delta$ while it could ``move on'' in $\Delta^+$. The set of all paths $P:[0,p]\to\Delta^+\cup\Delta$ admissible for $h$ is endowed with a partial order. For two paths $P_1:[0,p_1]\to\Delta^+\cup\Delta$, $P_2:[0,p_2]\to\Delta^+\cup\Delta$ admissible for $h$, we say that $P_1\preceq P_2$ if $p_1\leq p_2$ and $P_1(\tau) = P_2(\tau)$ for every $\tau\in[0,p_1]$.

In what follows, we focus on paths that are admissible for the homotopy induced on the function $f_{c(\tau)}^*$ by a continuous path $c:[0,1]\to\textrm{Reg}(f)$. With a slight abuse of notation, we talk about paths admissible for $c$.

\begin{proposition}\label{uniquePath}
Let $c:[0, 1]\to\textrm{Reg}(f)$ be a continuous path with $c(0) = (a,b)$. For every proper cornerpoint $X\in\dgm(f_{a,b}^*)$, a unique path $P:[0,1]\to\Delta^+\cup\Delta$ admissible for $c$ exists, such that
$P(0) = X$.
\end{proposition}

\begin{proof}\label{proofprop_uniquepath}

For every real number $\alpha\geq 0$, consider the property
\begin{center}
\begin{tabular}{p{0.95\textwidth}}
($\ast$) a path $P_{\alpha}:[0,\alpha]\to\Delta^+\cup\Delta$ admissible for $c$ exists, with $P_{\alpha}(0)=X$.
\end{tabular}
\end{center}
Define the set $A=\{\alpha\in[0,1]:\ \textrm{property ($\ast$) holds}\}$. $A$ is non-empty, since $0\in A$. Set $\bar\alpha = \sup A$. We need to show that $\bar\alpha\in A$. First, let $(\alpha_n)$ be a non-decreasing sequence of numbers of $A$ converging to $\bar\alpha$. Since $\alpha_n\in A$, for each $n$ there is a path $P_n:[0,\alpha_n]\to\Delta^+\cup\Delta$ admissible for $c$ with $P_n(0)=X$. By the Hausdorff maximal principle, we can consider a maximal chain of paths $P_{n}$ and define a function $P_{\bar\alpha}': [0,\bar\alpha)\to\Delta^+\cup\Delta$ by setting $P_{\bar\alpha}'(\tau)=P_{n}(\tau)$ for any $P_{n}$ in the maximal chain whose domain contains $\tau$. 

In particular, the function $P_{\bar\alpha}'$ is such that $P_{\bar\alpha}'(0)=X$ and $P_{\bar\alpha}'(\tau)\in\dgm(f_{c(\tau)}^*)$ for all $\tau\in [0,\bar\alpha)$. However, to prove that $\bar\alpha\in A$ we still need to show that $P_{\bar\alpha}'$ can be continuously extended to the point $\bar\alpha$. The localization of cornerpoints~\cite[Prop.~3.8]{CeDi*13} implies that, possibly by extracting a convergent subsequence, we can assume that the limit $\lim_n P_n(\alpha_n)=\lim_n P_{\bar\alpha}'(\alpha_n)$ exists. By the 1-dimensional Stability Theorem~\cite[Thm.~3.13]{CeDi*13}, we have that $\lim_n P_{\bar\alpha}'(\alpha_n)\in\dgm(f_{c(\bar\alpha)}^*)$. Now, the function $P_{\bar\alpha}'$ can be extended to a path $P_{\bar\alpha}:[0,\bar\alpha]\to\Delta^+\cup\Delta$ by setting $P_{\bar\alpha}(\bar\alpha)=\lim_n P_{\bar\alpha}'(\alpha_n)$. It is easy to check that $P_{\bar\alpha}$ is admissible for $c$, and hence $\bar\alpha\in A$.

Last, we prove by contradiction that $\bar\alpha=1$. Suppose that $\bar\alpha<1$. If $P_{\bar\alpha}(\bar\alpha)\not\in\Delta$, again by the 1-dimensional Stability Theorem and the fact that $c(\bar\alpha)\in\textrm{Reg}(f)$, for any sufficiently small $\varepsilon > 0$ we could take a real number $\eta >0$ such that there is exactly one proper cornerpoint $X'(\tau)\in\dgm(f_{c(\tau)}^*)$ with $d(X'(\tau),P_{\bar\alpha}(\bar\alpha))\leq\varepsilon$ for every $\tau$ with $\bar\alpha\leq\tau\leq\bar\alpha+\eta$. By setting $P_{\bar\alpha}(\tau)=X'(\tau)$ for every such $\tau$, we would get a continuous path that extends $P_{\bar\alpha}$ to the interval $[0,\bar\alpha+\eta)$. We could work similarly also in case $P_{\bar\alpha}(\bar\alpha)\in\Delta$. Indeed, our working assumption (2., Section~\ref{workAss}) implies that arbitrarily close to $P_{\bar\alpha}(\bar\alpha)$ we could find at most one proper cornerpoint $X'(\tau)$, for all $\tau$ with $\bar\alpha\leq\tau\leq\bar\alpha+\eta$ and $\eta$ sufficiently small, to be used to extend $P_{\bar\alpha}$. If there is no such a proper cornerpoint, $P_{\bar\alpha}$ could be extended by setting $P_{\bar\alpha}(\tau)=P_{\bar\alpha}(\bar\alpha)$ for the same values of $\tau$. In any case, we would get a contradiction of our assumption that $\bar\alpha=\sup A$.

We now show that there is a unique path $P:[0,1]\to\Delta^+\cup\Delta$ that is admissible for $c$ and starts at $X$. Assume that another path $P':[0,1]\to\Delta^+\cup\Delta$ admissible for $c$ exists, with $X=P(0)=P'(0)$. Denote by $\bar\tau$ the greatest value for which $P(\tau)=P'(\tau)$ for all $\tau\in[0,\bar\tau]$. Since $P$ differs from $P'$, $\bar \tau<1$. By the 1-dimensional Stability Theorem, if $P(\bar\tau)\not\in\Delta$ then $P(\bar\tau)$ is a proper cornerpoint of $\dgm(f_{c(\bar\tau)}^*)$ with multiplicity strictly greater than 1, against our assumption that $c(\tau)\in\textrm{Reg}(f)$ for all $\tau\in[0,1]$. If $P(\bar\tau)\in\Delta$ then $P$ and $P'$ contradict the definition of admissible path for $c$, because of our working assumption (2., Section~\ref{workAss}). Therefore, the path $P$ must be unique.
\end{proof}
We say that \emph{$c$ transports $X$ to $X'=P(1)$ with respect to $f$}. Now, we need to define the concept of transporting a matching along a path $c: [0,1]\to\textrm{Reg}(f)\cap\textrm{Reg}(g)$ with $c(0) = (a,b)$. Suppose that $\sigma_{(a,b)}(X)=Y$. Let $\sigma_{a,b}$ be a matching between $\dgm(f_{a,b}^*)$ and $\dgm(g_{a,b}^*)$, with $(a,b)$ an element of $\textrm{Reg}(f)\cap\textrm{Reg}(g)$. We can naturally associate to $\sigma_{a,b}$ a matching $\sigma_{c(1)}:\dgm(f_{c(1)}^*)\to\dgm(g_{c(1)}^*)$. We set  $\sigma_{c(1)} (X') = Y'$ if and only if $c$ transports $X$ to $X'$ with respect to $f$ and $Y$ to $Y'$ with respect to $g$. We also say that \emph{$c$ transports $\sigma_{a,b}$ to $\sigma_{c(1)}$ along $c$ with respect to the pair $(f,g)$}.

Following the same line of proof of Proposition~\ref{uniquePath}, we can also prove the following result.

\begin{proposition}\label{propGs}
Let $G(s,x) := (1-s)\cdot\varphi(x) + s\cdot\psi(x)$ be a homotopy between $\varphi,\psi:M\to\R$. Then, for every $X\in\dgm(\varphi)$ of multiplicity 1 and every sufficiently small $\varepsilon >0$, a unique path $P:[0,\varepsilon]\to\Delta^+\cup\Delta$ exists, such that $P(0) = X$ and $P$ is admissible for $G$ at any $\tau\in[0,\varepsilon]$.\end{proposition}
We say that \emph{$G$ transports $X$ to $X'$}.

We are now ready to introduce the coherent 2-dimensional matching distance.

\begin{definition}\label{CD}
Choose a point $(a,b)\in\textrm{Reg}(f)\cap\textrm{Reg}(g)$. Let $\Gamma$ be the set of all continuous paths $c:[0,1]\to Reg(f)\cap Reg(g)$ with $c(0)=(a,b)$. Let $S$ be the set of all matchings $\sigma:\dgm(f_{c(0)}^*)\to\dgm(g_{c(0)}^*)$. For every $\sigma\in S$ and every $c\in\Gamma$, the symbol $T^{(f,g)}_c(\sigma)$ will denote the matching obtained by transporting $\sigma$ along $c$ with respect to the pair $(f,g)$. We define the \emph{coherent 2-dimensional matching distance} $CD_{match}(\beta_f,\beta_g)$ as 
\begin{equation}
CD_{match}(\beta_f,\beta_g)=\max\left\{
\min_{\sigma\in S} \sup_{c\in \Gamma}\mbox{cost}\left(T^{(f,g)}_c(\sigma)\right),\gamma_\infty\right\},
\end{equation}
with 
\begin{itemize}
\item $\textrm{cost}\left(T^{(f,g)}_c(\sigma)\right)$ the cost of the matching $T^{(f,g)}_c(\sigma)$ with respect to the max-norm;
\item $\gamma_\infty$ the maximum distance between the cornerpoint at infinity of $f_{a,b}^*$ and the cornerpoint at infinity of $g_{a,b}^*$ varying $(a,b)$ in $P(\Lambda)$ for degrees $0$ and $m$, and $0$ for the other degrees.
\end{itemize}
\end{definition}
The following statements hold, proving that $CD_{match}$ is well-defined and satisfies the properties of a pseudo-distance.  

\begin{proposition}\label{prop1}
The definition of  $CD_{match}(\beta_f,\beta_g)$ does not depend on the choice of the point $(a,b)\in\textrm{Reg}(f)\cap\textrm{Reg}(g)$.
\end{proposition}
\begin{proof}\label{proofprop1}
Let us choose another basepoint $(a',b')\in\textrm{Reg}(f)\cap\textrm{Reg}(g)$. We can take a path $c'\in\Gamma$ with $c'(1)=(a',b')$. It is sufficient to observe that $T^{(f,g)}_c(\sigma)=T^{(f,g)}_{c\ast c'^{-1}}\left(T^{(f,g)}_{c'}(\sigma)\right)$
for every $\sigma\in S$ and every $c\in \Gamma$, where $\ast$ denotes the concatenation of paths and $c'^{-1}$ is the inverse path of $c'$, i.e., $c'^{-1}(t):=c'(1-t)$.
\end{proof}
\begin{proposition}\label{prop2}
$CD_{match}(\beta_f,\beta_g)$ is a pseudo-distance.
\end{proposition}
\begin{proof}\label{proofprop2}
It is sufficient to observe that if two matchings $\sigma:\dgm(f_{a,b}^*)\to\dgm(g_{a,b}^*)$, $\tau:\dgm(g_{a,b}^*)\to\dgm(h_{a,b}^*)$ are given, then $T^{(f,h)}_c(\tau\circ\sigma)=T^{(g,h)}_c(\tau)\circ T^{(f,g)}_c(\sigma)$  for every $c\in\Gamma$ taking values in $\textrm{Reg}(f)\cap\textrm{Reg}(g)\cap\textrm{Reg}(h)$. This implies that $\mbox{cost}\left(T^{(f,h)}_c(\tau\circ\sigma)\right)\le\mbox{cost}\left(T^{(g,h)}_c(\tau)\right)+\mbox{cost}\left(T^{(f,g)}_c(\sigma)\right)$. Hence the triangle inequality follows.
\end{proof}
The next result shows the stability of the coherent 2-dimensional matching distance.

\begin{theorem}\label{thmstability}
It holds that $CD_{match}(\beta_f,\beta_g)\le\|f-g\|_\infty$.
\end{theorem}
\begin{proof}\label{proofprop3}
First of all we recall that for every $(a,b)\in P(\Lambda)$ and every $x\in M$, we have
$$
\begin{array}{lll}
|f_{a,b}^*(x)-g_{a,b}^*(x)| & = & \min\{a,1-a\}\cdot\left|f_{a,b}(x)-g_{a,b}(x)\right| \leq \\
%& & \min\{a,1-a\}\cdot\max\left\{\left|\frac{f_1(p)-b}{a}-\frac{{g}_1(p)-b}{a}\right|, \left|\frac{f_2(p)+b}{1-a}-\frac{{g}_2(p)+b}{1-a}\right|\right\} = \\
& & \min\{a,1-a\}\cdot\max\left\{\left|\frac{f_1(x)-{g}_1(x)}{a}\right|, \left|\frac{f_2(x)-{g}_2(x)}{1-a}\right|\right\} \leq \\
& & \max\left\{\left|f_1(x)-{g}_1(x)\right|, \left|f_2(x)-{g}_2(x)\right|\right\},
\end{array}
$$
in other words, $||f_{a,b}^*- g_{a,b}^*||_\infty\le \|f-g\|_\infty$. Let us consider the closed set $C$ obtained from $P(\Lambda)$ by taking away a union of small balls of radius $\delta$ around the singular pairs for $f$. Let $\varepsilon$ be the infimum for $(a,b)\in C$ of the minimum distance between any two proper cornerpoints of $\dgm(f_{a,b}^*)$. Note that, setting $K=\max_{x\in M}\{\|f(p)\|_{\infty},\|g(p)\|_{\infty}\}$, the computation of $\varepsilon$ can be accomplished on the compact set $C'=\{(a,b)\in C:|b|\leq K\}$. Indeed, for all $(a,b)\in C$ with $|b|>K$, we have that either $\dgm(f_{a,b}^*)=\dgm(f_1^*)$ or $\dgm(f_{a,b}^*)=\dgm(f_2^*)$ (analogously for $g$). Hence, by recalling our working assumption (2., Section~\ref{workAss}) and by construction of $C'$, we have that $\varepsilon >0$. 

Let $\varepsilon'$ be the infimum for $(a,b)\in C$ of the minimum distance between any two points of $\dgm(g_{a,b}^*)$. From the fact that $||f_{a,b}^*- g_{a,b}^*||_\infty\le \|f-g\|_\infty$, we know that when we change $f$ according to the homotopy $G_s:=(1-s)\cdot f+s\cdot g$, each point in the normalized persistence diagram moves by a distance not greater than  $\|f-G_s\|_\infty\le \|f-g\|_\infty$, because of the stability of 2-dimensional persistent Betti numbers under the multidimensional matching distance \cite[Thm.~4.2]{CeDi*13}. Hence we have that $\varepsilon'\ge\varepsilon - 2\|f-G_s\|_\infty\ge\varepsilon - 2\|f-g\|_\infty$.
Therefore, if $\|f-g\|_\infty$ is  small enough, we have that $\varepsilon>0$ implies $\varepsilon'>0$.

As a consequence, if $\|f-g\|_\infty$ is  small enough, for each $(a,b)\in C$ we can consider the unique matching $\sigma_{a,b}:\dgm(f_{a,b}^*)\to \dgm(g_{a,b}^*)$ obtained by changing the identity $id_{a,b}:\dgm(f_{a,b}^*)\to \dgm(f_{a,b}^*)$ according to the change of persistence diagrams induced by the homotopy $G_s^*:=(1-s)\cdot f_{a,b}^*+s\cdot g_{a,b}^*$ (see Proposition~\ref{propGs}). Formally, for each proper cornerpoint $X$ of $\dgm(f_{a,b}^*)$, we set  $\sigma_{a,b} (X) = X'$ if and only if the homotopy $G_s^*$ transports $X$ to the cornerpoint $X'$ of $\dgm(g_{a,b}^*)$. 

We have that $\cost (\sigma_{a,b})\leq\|f_{a,b}^*-g_{a,b}^*\|_\infty\le \|f-g\|_\infty$. If $\hat c$ is a continuous path from $(\bar a,\bar b)$ to $(a,b)$ with $\hat c([0,1])\subset C$, it is easy to see that the function $\sigma_{\hat c(t)}$ in the variable $t$ describes the transport 
of $\sigma_{\bar a,\bar b}$ to $\sigma_{a,b}$ made by $\hat c$ with respect to the pair $(f,g)$, because $\sigma_{\hat c(t)}$ depends continuously on $t$.

Now, let $c$ be a continuous path from $(\bar a,\bar b)$ to $(a,b)$ with $c([0,1])\subset\textrm{Reg}(f)\cap\textrm{Reg}(g)$. The image of $c$ may be not contained in $C$, but if $(\bar a,\bar b),(a,b)\in C$, $c$ is relative homotopic to a path $\hat c:[0,1]\to C$ from $(\bar a,\bar b)$ to $(a,b)$. This path $\hat c$ transports $\sigma_{\bar a,\bar b}$ to $\sigma_{a,b}$, too. 

Hence we have that $\mbox{cost}\left(T^{(f,g)}_c(\sigma_{\bar a,\bar b})\right)=\mbox{cost}\left(T^{(f,g)}_{\hat c}(\sigma_{\bar a,\bar b})\right)=\mbox{cost}\left(\sigma_{a, b}\right)\le\|f-g\|_\infty$. From this and from the fact that we can choose an arbitrarily small $\delta$, it follows that $\min_{\sigma\in S}\sup_{c\in\Gamma}\mbox{cost}\left(T^{(f,g)}_c(\sigma_{\bar a,\bar b})\right)
\le \|f-g\|_\infty$ for every $c\in\Gamma$.
\end{proof}
We conclude by observing that the definition of the coherent 2-dimensional matching distance immediately implies that it is not less informative than the usual 2-dimensional matching distance. Formally,

\begin{proposition}\label{propD_CD}
$D_{match}(\beta_f,\beta_g)\le CD_{match}(\beta_f,\beta_g)$.
\end{proposition}

\section*{Conclusions}
In this contribution we have introduced the notion of coherent 2-dimensional matching distance. Similarly to the classical 2-dimensional matching distance, it is based on defining a suitable family of filtrations associated with lines having a positive slope; however, this new distance only considers matchings that change coherently with respect to the filtrations which are taken into account. 

Through the paper we formally proved that the coherent matching distance distance between 2-dimensional Persistent Betti numbers is well-defined, stable and does not loose discriminative information with respect to the usual 2-dimensional matching distance. We believe that these first results make the coherent matching distance deserving further investigation, such as extending the study of its theoretical properties as well as exploring how to develop computational techniques for its practical evaluation. 

{\small
\bibliographystyle{splncs03}
\bibliography{coherentMatchingDistance_arXiv}
}

\end{document}